\documentclass[runningheads]{llncs}
\usepackage{graphicx}
\usepackage[linesnumbered,ruled]{algorithm2e}
\usepackage{setspace}
\usepackage{amsmath, amsfonts} 
\usepackage{comment}
\usepackage{hyperref}

\DeclareMathOperator*{\argmin}{arg\,min}
\DeclareMathOperator*{\argmax}{arg\,max}

\newcommand{\eqdef}{\mathrel{\mathop:}=}

\DeclareMathOperator{\NP}{\NP}
\DeclareMathOperator{\MC}{MC}
\DeclareMathOperator{\QST}{QST}
\DeclareMathOperator{\MWST}{MWST}
\DeclareMathOperator{\BCDS}{BCDS}
\DeclareMathOperator{\OPT}{OPT}
\DeclareMathOperator{\BEVD}{BEVD}

\DeclareMathOperator{\BEVDC}{BEVD_C}
\DeclareMathOperator{\PCDS}{PCDS}
\begin{document}
\title{Improved Budgeted Connected Domination and Budgeted Edge-Vertex Domination\thanks{This work was partially supported by the Special Account for Research Grants (ELKE) of the National and Kapodistrian University of Athens (NKUA).}}

\titlerunning{Improved Budgeted Connected Domination}
%
\author{Ioannis Lamprou \and
Ioannis Sigalas \and
Vassilis Zissimopoulos}
\authorrunning{I. Lamprou et al.}
%
\institute{Department of Informatics and Telecommunications,\\ National and Kapodistrian University of Athens, Greece\\
\email{\{ilamprou, sigalasi, vassilis\}@di.uoa.gr}
}

\maketitle              

\begin{abstract}
We consider the \emph{Budgeted} version of the classical \emph{Connected Dominating Set} problem (BCDS).
Given a graph $G$ and a budget $k$, we seek a connected subset of at most $k$ vertices maximizing the number of dominated vertices in $G$. 
We improve over the previous $(1-1/e)/13$ approximation in [Khuller, Purohit, and Sarpatwar,\ \emph{SODA 2014}] by introducing a new method for performing tree decompositions in the analysis of the last part of the algorithm.
This new approach provides a $(1-1/e)/12$ approximation guarantee.
By generalizing the analysis of the first part of the algorithm, we are able to modify it appropriately and obtain a further improvement to $(1-e^{-7/8})/11$.
On the other hand, we prove a $(1-1/e+\epsilon)$ inapproximability bound, for any $\epsilon > 0$.

We also examine the \emph{edge-vertex domination} variant, where an edge dominates its endpoints and all vertices neighboring them.
In \emph{Budgeted Edge-Vertex Domination} (BEVD), we are given a graph $G$, and a budget $k$, and we seek a, not necessarily connected, subset of $k$ edges such that the number of dominated vertices in $G$ is maximized.
We prove there exists a $(1-1/e)$-approximation algorithm.
Also, for any $\epsilon > 0$, we present a $(1-1/e+\epsilon)$-inapproximability result by a gap-preserving reduction from the \emph{maximum coverage} problem.
Finally, we examine the ``dual'' \emph{Partial Edge-Vertex Domination} (PEVD) problem, where a graph $G$ and a quota $n'$ are given.
The goal is to select a minimum-size set of edges to dominate at least $n'$ vertices in $G$.
In this case, we present a $H(n')$-approximation algorithm by a reduction to the \emph{partial cover} problem.

\keywords{Approximation \and Budget \and Partial \and Connected Domination \and Edge-Vertex Domination}
\end{abstract}
\section{Introduction} 

The problem of vertices dominating vertices in a graph is very common and has been extensively studied in graph theory and combinatorial optimization literature.
In the classical definition, a dominating set is a subset of vertices such that each vertex is either a member of the subset or adjacent to a member of the subset.
Intuitively, a dominating set provides a skeleton for the placement of resources, such that any network node is within immediate reach to them. 

However, as it is often the case, there are constraints on the amount of resources available for placement, e.g., due to financial or other management reasons.
That is, we are limited to a budget of $k$ resources to be placed on network nodes.
The optimization goal is to place the available resources suitably, such that the number of network nodes they dominate is maximized.
This problem is known in literature as the \emph{Budgeted Dominating Set} problem.

Budgeted domination has applications especially in 
ad-hoc wireless (sensor) networks.
In this setting, a set of network nodes needs to be identified as the virtual backbone of the network, that is, the structure responsible for routing and packet forwarding.
To achieve these tasks, nodes in the backbone must be able to communicate with each other, i.e., form a \textit{connected} set of vertices in the graph capturing the topology of their communication ranges.
The resulting optimization problem is the \emph{Budgeted Connected Dominating Set} (BCDS) problem.
In this paper, we study BCDS and present an improved guarantee over the previous state of the art \cite{Khuller}. 

Besides BCDS, we examine other problems where graph edges are selected as dominators. 
The concept of edges dominating adjacent edges has been well-considered in literature; e.g., see \cite{Horton,Yannakakis} for some preliminary results.
An example application is in network tomography where probes need to be placed to monitor the health of network links \cite{Kumar}.

In this paper, we consider cases where resources must be positioned on the links of a network to dominate network nodes.
For instance, consider a power system where a limited number of static var compensators need to be placed on transmission lines' midpoints to locate faults affecting a big proportion of buses \cite{Khoa}.
Another example is to identify a limited-size set of friendships, modeled as graph edges, having a big impact in terms of neighborhood in a social network.

More formally, the notion in consideration is \textit{edge-vertex domination}, where an edge dominates its endpoints and any vertices adjacent to its endpoints.
We examine the (in)approximability of \emph{Budgeted Edge-Vertex Domination} (BEVD), where we seek a, not necessarily connected, set of $k$ (budget) edges dominating as many vertices as possible.
If the edge set is required to be connected, we show that the problem essentially matches BCDS.
Finally, we consider the related \emph{Partial Edge-Vertex Domination} (PEVD) problem: a quota of vertices needs to be dominated by utilizing the minimum number of edges possible.

\subsection{Related Work}
Finding a minimum-size connected set of vertices dominating the whole graph is a classical NP-hard problem.
In \cite{Guha}, Guha and Khuller proposed a $\ln\Delta + 3$ approximation algorithm, which is (up to constant factors) the best possible, since the problem is hard to approximate within a factor of $(1-\epsilon)\log n$ \cite{Feige}.
For a bigger picture of the research landscape, in \cite{Du}, many connected domination results for special graph classes and other applications are surveyed.

In \cite{Miyano}, vertex-vertex and edge-edge budgeted domination are considered.
For vertex-vertex, matching upper and lower bounds of $(1-1/e)$ are given, whereas, for edge-edge, a $(1-1/e)$ approximation and a $1303/1304 + \epsilon$ hardness are proved.

In the connected case, budgeted and partial versions of domination have their origins in wireless sensor networking \cite{Liu,Wang}, where a network backbone with good qualities needs to be determined, which must either be limited in resources and/or cover a big-enough proportion of the network. 
The first, and thus far state of the art, results for the budgeted and partial cases in general graphs appear in \cite{Khuller}, where a $(1-1/e)/13$-approximation, respectively an $O(\ln\Delta)$-approximation, is proved for the budgeted, respectively partial, case. 
Other works have followed in particular settings.
For example, in \cite{Liu-Wang}, a constant factor approximation algorithm for partial connected domination on a superset of unit disk graphs, namely growth-bounded graphs, is proposed.
Their result translates to a $27$-approximation guarantee on unit disk graphs. 

Regarding edge-vertex domination, the graph-theoretic notion was introduced in \cite{Peters}, together with the complementary case of vertex-edge domination, where a vertex dominates all edges incident to it or to a neighbor of it.
Some complexity and algorithmic results about the minimal size of an edge-vertex, respectively vertex-edge, dominating set appear in \cite{Lewis}.
More recently, some vertex-edge domination open questions posed in \cite{Lewis} were answered in \cite{Boutrig}.
In \cite{Venkat}, an improved bound on the edge-vertex domination number of trees was proved.
Except for the vertex-edge and edge-vertex variants, a \emph{mixed} domination variant has been introduced \cite{Sampa}, where a minimal subset of both vertices and edges need to be selected so that each vertex/edge of the graph is incident/adjacent to a vertex/edge in the subset.
Recent example works in this topic study the problem in special graph classes like trees, cacti, and split graphs \cite{Lan,Zhao}.

\subsection{Our Results}

In Section~\ref{sec:prel}, we present preliminary notions and formally define the problems.

In Section~\ref{sec:bcds}, we examine the Budgeted Connected Dominating Set (BCDS) problem, see Definition~\ref{def:bcds}, where a connected subset of budget vertices needs to dominate as many vertices as possible.
By introducing a new tree decomposition technique in Subsection~\ref{sec:eligible}, we prove a $(1-1/e)/12 \simeq 0.05267$ approximation, in Theorem~\ref{thm:1/12}, which improves over the previous best known $(1-1/e)/13$ guarantee \cite{Khuller}.
(We note the same guarantee has recently been achieved independently in \cite{Khuller-journal}.)
We further improve the ratio to $(1-e^{-7/8})/11 \simeq 0.05301$ (Theorem~\ref{thm:general-c}) by generalizing the first part of the analysis in \cite{Khuller} and then modifying the proposed algorithm accordingly in Subsection~\ref{sec:ck}.
On the negative side, for any $\epsilon > 0$, we show a first $(1-1/e+\epsilon)$ inapproximability bound; see Theorem~\ref{thm:inapprBCDS}.

In Section~\ref{sec:edge-vertex}, we consider edge-vertex domination, where a, not necessarily connected, subset of edges dominates adjacent vertices. 
If the set of edges is also required to be connected, then the problems essentially reduce to the standard vertex-vertex budgeted/partial dominating set problems; see Proposition~\ref{prop:edge-vertex-connected}.
In Subsection~\ref{sec:bevd}, we prove there is a $(1-1/e)$-approximation algorithm (Theorem~\ref{thm:bevd-approx}).
This is the best possible since we prove an $(1-1/e+\epsilon)$ inapproximability lower bound, for any $\epsilon > 0$, see Theorem~\ref{thm:inappr}.
In Subsection~\ref{sec:pevd}, we consider the problem of Partial Edge-Vertex Domination.
In Theorem~\ref{thm:partial}, we prove that, in the general case, there exists an $H(n')$-approximation, where $H(\cdot)$ is the Harmonic number and $n'$ is the number of vertices requested to be dominated. 
To do so, we employ a reduction to a partial version of the classical \emph{Set Cover} problem.

Finally, in Section~\ref{sec:conclusion}, we give some concluding remarks. 

\section{Preliminaries}\label{sec:prel}
A graph $G$ is denoted as a pair $(V(G), E(G))$ (or simply $(V,E)$) of the vertices and edges of $G$.
The graphs considered are simple (neither loops nor multi-edges are allowed), connected and undirected.
Besides the aforementioned, no assumptions are made on the topology of the input graphs.

Two vertices $u,v \in V$ connected by an edge, denoted $(u,v)$ or equivalently $(v,u)$, are called \emph{adjacent} or \emph{neighboring}.
The \emph{open neighborhood} of a vertex $v \in V$ is defined as $N(v) = \{u \in V: (v, u) \in E\}$, while the \emph{closed neighborhood} is defined as $N[v] = \{v\} \cup N(v)$.
For a subset of vertices $S \subseteq V(G)$, we expand the above definitions to $N(S) = \bigcup_{v \in S} N(v) \setminus S$ and $N[S] = N(S) \cup S$.

The degree of a vertex $v \in V$ is defined as $d(v) = |N(v)|$.
The minimum, resp.\ maximum, degree of $G$ is denoted by $\delta = \min_{v \in V} d(v)$, resp.\ $\Delta = \max_{v \in V} d(v)$.

Let us now consider the neighborhood of edges in terms of vertices.
Given an edge $e = (v,u) \in E$, let $I(e) = \{v, u\}$ stand for the set containing its two incident vertices.
We define the \emph{neighborhood of an edge} $e$ as $N[e] = \bigcup_{v \in I(e)} N[v]$.
For a set of edges $E' \subseteq E$, we define $V(E') = \{v \in V \;|\; \exists e \in E'\text{ such that } v\in I(e)\}$. 
Then, we define the \emph{edge-set neighborhood} as $N[E'] = N[V(E')]$.
Here, we focus on a closed neighborhood definition, since it captures the number of vertices incident or adjacent to a set of edges in the standard edge-vertex domination paradigm (Definition~8 in \cite{Lewis}; originally introduced in \cite{Peters}).
That is, we say that a set of edges $E'$ \emph{dominates} $N[E']$.

Let us now proceed to formally define the problems studied in this paper.

\begin{definition}[BUDGETED CONNECTED DOMINATING SET]\label{def:bcds}
	\newline Given a graph $G = (V,E)$ and an integer $k$, select a subset $S \subseteq V$, where $|S| \le k$, such that the subgraph induced by $S$ is connected and $|N[S]|$ is maximized.
\end{definition}

\begin{definition}[BUDGETED EDGE-VERTEX DOMINATION]\label{def:budget}
	\newline Given a graph $G=(V,E)$ and an integer $k$, select a subset $E' \subseteq E$, where $|E'| \le k$, such that $|N[E']|$ is maximized.
\end{definition}

\begin{definition}[PARTIAL EDGE-VERTEX DOMINATION]\label{def:partial}
	Given a graph $G=(V,E)$ and an integer $n'$, select a subset $E' \subseteq E$ of minimum size such that it holds $|N[E']| \ge n'$.
\end{definition}


\section{Budgeted Connected Dominating Set}\label{sec:bcds}

In this section, we consider the Budgeted Connected Dominating Set (BCDS) problem given in Definition~\ref{def:bcds}.
We initially present a summary of key aspects of the state of the art algorithm \cite{Khuller}, which achieves a $(1-1/e)/13$ approximation factor.
We then show how the analysis can be improved to achieve a $(1-1/e)/12$ guarantee via an alternative tree decomposition scheme; see Theorem~\ref{thm:1/12}.
Then, we generalize the analysis of the greedy procedure in order to modify a call within the state of the art algorithm.
This modification allows us to increase the approximation factor even further to $(1-e^{-7/8})/11$; see Corollary~\ref{cor:7/8}.
On the other hand, we conclude this section with a $(1-1/e+\epsilon)$, for any $\epsilon > 0$, inapproximability result; see Theorem~\ref{thm:inapprBCDS}.

\subsection{Previous Approach}
Khuller et al., see Algorithm~\ref{alg:khuller} (Algorithm 5.1 in~\cite{Khuller}), design the first constant factor approximation algorithm for BCDS with an approximation guarantee of $(1-1/e)/13$.
Their approach comprises three method calls:
(i) a call to an algorithm returning a greedy dominating set $D$ and its corresponding profit function $p$; see Algorithm~\ref{alg:greedy} (GDS),
(ii) a call to a $2$-approximation algorithm, which follows from \cite{Garg,Johnson}, for the \emph{Quota Steiner Tree} (QST) problem defined below, and
(iii) a call to a dynamic programming scheme $Best_k(\cdot)$ to determine the maximum-profit subtree of size at most $k$ within a bigger-size tree.

\begin{figure}[h]
	\scalebox{0.95}{
		\begin{algorithm}[H]
			\SetKwInOut{Input}{Input}
			\SetKwInOut{Output}{Output}
			\DontPrintSemicolon
			
			\Input{A graph $G = (V(G), E(G))$}
			\Output{A dominating set $D \subseteq V(G)$ and a profit function $p: V(G) \rightarrow \mathbb{N} \cup \{0\}$}
			
			$D \leftarrow \emptyset$\;
			$U \leftarrow V(G)$\;
			\ForEach{$\upsilon \in V(G)$}{
				$p(\upsilon) \leftarrow 0$\;
			}
			\While{$U \neq \emptyset$}{
				$w \leftarrow \argmax_{\upsilon \in V(G) \setminus D} |N_U(\upsilon)|$ \tcc*{$N_U(\upsilon) = N[\upsilon]\cap U$}
				$p(w) \leftarrow |N_U(w)|$\;
				$U \leftarrow U \setminus N_U(w)$\;
				$D \leftarrow D \cup \{w\}$\;
			}
			\textbf{return} $(D, p)$
			
			\caption{Greedy Dominating Set (GDS) \cite{Khuller}}
			\label{alg:greedy}
		\end{algorithm}
	}
\end{figure}

\begin{figure}[h]
	\scalebox{0.95}{
		\begin{algorithm}[H]
			\SetKwInOut{Input}{Input}
			\SetKwInOut{Output}{Output}
			\DontPrintSemicolon
			
			\Input{A graph $G = (V(G), E(G))$ and $k \in \mathbb{N}$}
			\Output{A tree $\tilde{T}$ on at most $k$ vertices}
			
			$(D, p) \leftarrow GDS(G)$\;
			$T \leftarrow QST(G, (1-1/e)\OPT, p)$\;
			$\tilde{T} \leftarrow  Best_k(T, p)$\;
			
			\textbf{return} $\tilde{T}$
			
			\caption{Greedy Profit Labeling Algorithm for BCDS \cite{Khuller}}
			\label{alg:khuller}
		\end{algorithm}
	}
\end{figure}

\begin{definition}[QUOTA STEINER TREE]
	Given a graph $G$, a vertex profit function $p: V(G) \rightarrow \mathbb{N} \cup \{0\}$, an edge cost function $c: E(G) \rightarrow \mathbb{N} \cup \{0\}$ and a quota $q \in \mathbb{N}$, find a subtree $T$ that minimizes $\sum_{e\in E(T)} c(e)$ subject to the condition $\sum_{v \in V(T)} p(v) \ge q$. 
\end{definition}

\begin{theorem}[Follows from results in \cite{Garg,Johnson}]\label{thm:QST}
	There is a $2$-approximation algorithm for QUOTA STEINER TREE.
\end{theorem}

In their analysis, Khuller et al.~\cite{Khuller} demonstrate that there exists a set $D' \subseteq D$ of size $k$ which dominates at least $(1-1/e)\OPT$ vertices, where $\OPT$ is the optimal number of dominated vertices achieved with a budget of $k$.
Furthermore, $D'$ can be connected by adding at most another $2k$ Steiner vertices, so giving a total of $3k$ vertices.
Then, it suffices to call the $2$-approximation algorithm for $\QST$, see line 2 in Algorithm~\ref{alg:khuller}, with profit function $p$ (returned by algorithm GDS at line 1), all edge costs equal to $1$ and quota equal to $(1-1/e)\OPT$. The value $\OPT$ can be guessed via a binary search between $k$ and $n$.
Overall, the returned tree has size at most $6k$ vertices and dominates at least $(1-1/e)\OPT$ vertices: a $(6, 1-1/e)$ bicriteria approximation is attained (Lemma~5.2 \cite{Khuller}).

As a final step ($Best_k(\cdot)$ at line 3), a dynamic programming approach is used to identify the best-profit subtree with at most $k$ vertices, such that the budget requirement is satisfied; see paragraph 5.2.2 in \cite{Khuller} for the relevant recurrences.
To obtain a true approximation guarantee for the final solution, the following tree decomposition lemma is used recursively to prove that, for a sufficiently large value of $k$, a tree of size $6k$ can be decomposed into $13$ trees; each of size at most $k$ (Lemma~5.4 \cite{Khuller}).

\begin{lemma}[Folklore
	]\label{lem:jordan-decomposition}
	Given any tree on $n$ vertices, we can decompose it into two trees (by replicating a single vertex) such that the smaller tree has at most $\lceil \frac{n}{2} \rceil$ vertices and the larger tree has at most $\lceil \frac{2n}{3} \rceil$ vertices.
\end{lemma}



\subsection{Improvement to Previous Approach: Eligible Trees}\label{sec:eligible}
An improvement to the analysis in \cite{Khuller} can be achieved by utilising a more refined tree decomposition (than the recursive application of Lemma~\ref{lem:jordan-decomposition}) to provide the approximation guarantee at the final step. 
To do so, we consider a tree decomposition scheme based on the notion of \emph{eligible trees} as introduced in \cite{Bermond}.

\begin{definition}[\cite{Bermond}]
	Given a directed tree $T = (V_T, E_T)$, an \emph{eligible subtree} $T'$ is a subtree of $T$ rooted at some vertex $i \in V_T$ such that the forest obtained by deleting the edges with both endpoints in $T'$, and then all the remaining vertices of degree $0$, consists of a single tree.
\end{definition}

Assuming $T'$ is an eligible subtree not identical to $T$, after deleting all edges with both endpoints in $T'$, the only vertex of $T'$ with degree strictly greater than $0$ is the root vertex of $T'$.
That is, like in Lemma~\ref{lem:jordan-decomposition}, a single vertex is replicated when removing $T'$ from $T$; see Figure~\ref{fig:eligible}.
The following lemma suggests that, for any tree, there exists an eligible subtree within some specific size range.


\begin{lemma}[Lemma~5~\cite{Bermond}]\label{lem:eligible}
	For each directed tree $T = (V_T, E_T)$, and for each $p \in [1,|V_T|]\cap\mathbb{N}$, there exists an eligible subtree $T'$ of $T$ such that $p/2 \le |V_{T'}|  \le p$.
\end{lemma}
\begin{figure}[h]
	\centering
	\includegraphics[scale=0.65]{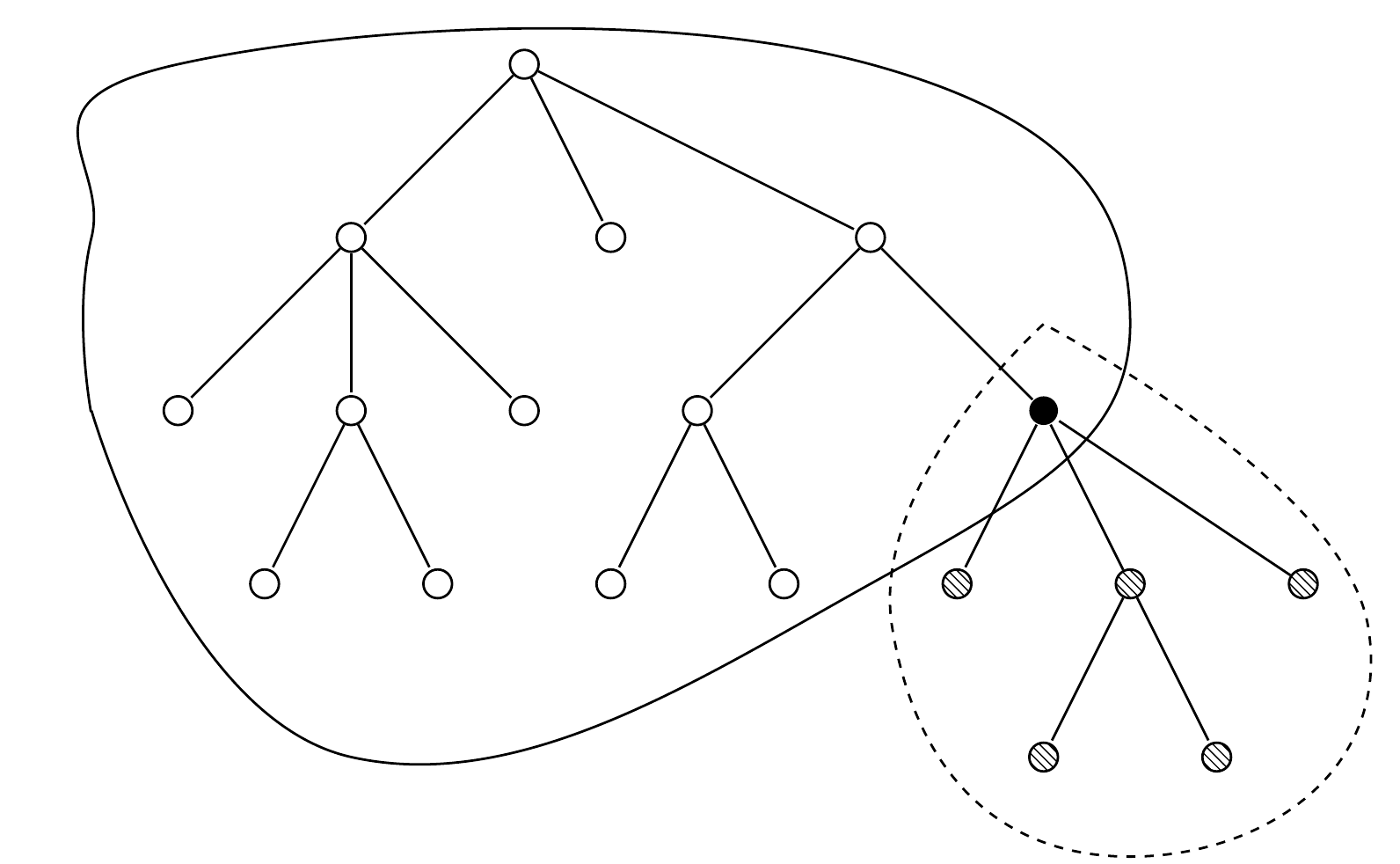}
	\caption{An example eligible subtree of size $6$ (enclosed within the dashed shape).
		After removing its edges and then all remaining vertices of degree $0$ (vertices with lines), 
		a single tree remains (enclosed within the solid shape). 
		A single vertex is replicated in both trees, the black vertex.
	}
	\label{fig:eligible}
\end{figure}
We can now proceed to employ the above lemma iteratively toward a decomposition scheme for the tree of size at most $6k$ returned by the Quota Steiner Tree call in Algorithm~\ref{alg:khuller}.

\begin{lemma}\label{lem:eligible-ak}
	Let $k$ be an integer. 
	Given any tree $T$ on $ak$ vertices, where $a \in \mathbb{N}$ is a constant, and $k \ge 4a-2$, we can decompose it into $2a$ subtrees each on at most $k$ vertices.
\end{lemma}	
\begin{proof}
	To make $T$ directed, we orient its edges away from some arbitrary vertex picked as the root.
	Now, we iteratively apply Lemma~\ref{lem:eligible} with $p = k$, until we are left with a tree on at most $k$ vertices.
	
	First, let us show that after $i$ iterations, the remaining tree has at most $ak - i\cdot (k/2 - 1)$ vertices.
	At the first iteration, there exists an eligible subtree $T'_1$ such that $k/2 \le |V_{T'_1}| \le k$.
	After removing it from $T_1 \eqdef T$ we are left with $T_2$ of size $|V_{T_1}| - (|V_{T'_1}| - 1)$, since the root of $T'_1$ remains in $T_1$.
	Hence, $|V_{T_1}| \le ak - (k/2 - 1)$, since $k/2 \le |V_{T'_1}|$.
	Assume that after $i$ iterations of the above procedure, it holds for the remaining tree $T_{i+1}$ that $k < |V_{T_{i+1}}| \le ak - i\cdot (k/2 - 1)$.
	We inductively apply Lemma~\ref{lem:eligible} with $p = k$ and get an eligible subtree $T'_{i+1}$.
	Removing $T'_{i+1}$ from $T_{i+1}$, we get $T_{i+2}$, where $|V_{T_{i+2}}| = |V_{T_{i+1}}| - (|V_{T'_{i+1}}| - 1) \le
	ak - i\cdot (k/2 - 1) - (k/2 - 1) = ak - (i+1)\cdot (k/2 - 1)$.
	
	We proved that, after $i$ removals of eligible subtrees from the original tree, for the remaining tree $T_{i+1}$ it holds
	$|V_{T_{i+1}}| \le ak - i\cdot (k/2 - 1)$.
	For $i = 2a-1$, we get $|V_{T_{2a}}| \le ak - (2a-1)\cdot (k/2 - 1) = ak - ak + 2a + k/2 - 1 = k/2 + 2a - 1$, which is at most $k$ for a sufficiently large value of $k$, i.e., $k \ge 4a-2$.
	Overall, the original tree $T_1$ has been decomposed into $2a$ trees: $T'_1, T'_2, \ldots, T'_{2a-1}$ and $T_{2a}$, each of which has at most $k$ vertices.
\qed\end{proof}

\begin{theorem}\label{thm:1/12}
	Algorithm~\ref{alg:khuller} is a $(1-1/e)/12$ approximation for BCDS.
\end{theorem}
\begin{proof}
	By Lemma 5.2 in \cite{Khuller}, there is a bicriteria $(6k, 1-1/e)$ approximation for BCDS. 
	That is, line 2 in Algorithm~\ref{alg:khuller} returns a tree of size at most $6k$ which dominates at least $(1-1/e)\OPT$ vertices.
	By Lemma~\ref{lem:eligible-ak}, for $a = 6$, a tree on $6k$ vertices can be decomposed into $12$ subtrees on at most $k$ vertices.
	To obtain the best subtree on at most $k$ vertices, we run the dynamic programming procedure (subsection 5.2.2 \cite{Khuller}).
	Let $T^*$ be the returned subtree.
	Then, it holds
	$ p(T^*) \ge \frac{1}{12} \sum_{i = 1}^{12} p(T_i) \ge \frac{1}{12} p(T) \ge \frac{1}{12}(1-1/e)\OPT$.
\qed\end{proof}

\subsection{An Improved Modified Algorithm}\label{sec:ck}

In the following proof, we generalize the analysis given in Lemma 5.1 \cite{Khuller} regarding the existence of a greedily selected set (of at most $k$ vertices) with a good intersection to the (neighborhood of the) optimal solution. Below, let $D$ and $p$ refer to the dominating set and profit function returned by GDS (line 1 in Algorithm~\ref{alg:khuller}).

\begin{lemma}\label{lem:analysis}
	There exists a set $D' \subseteq D$, $|D'|\le \lceil ck \rceil$, for some constant $0 < c \le 1$, such that $p(D') \ge (1-e^{-c})\OPT$.
	Furthermore, $D'$ can be connected using at most another $k + \lceil ck \rceil$ Steiner vertices.
\end{lemma}
\begin{proof}
	We define the layers $L_1, L_2, L_3$ as follows.
	$L_1$ contains the (at most $k$) vertices of an optimal BCDS solution.
	Let $L_2 = N(L_1)$, meaning that the optimal number of dominated vertices is $\OPT = |L_1 \cup L_2|$.
	Also, let $L_3 = N(L_2)\setminus L_1$ and $R = V \setminus (L_1 \cup L_2 \cup L_3)$, where $R$ denotes the remaining vertices, i.e., those outside the three layers $L_1, L_2, L_3$.
	Let us now consider the intersection of these layers with the greedy dominating set $D$ returned by GDS (Algorithm~\ref{alg:greedy}).
	Let $L'_i = D \cap L_i$ for $i = 1,2,3$ and $D' = \{v_1, v_2, \ldots, v_\lambda\}$ denote the first $\lambda = \lceil ck \rceil$ vertices from $L'_1 \cup L'_2 \cup L'_3$ in the order selected by the greedy algorithm. 
	In order to bound the total profit in $D'$, we define $g_i = \sum_{\mu=1}^i p(v_\mu)$ as the profit we gain from the first $i$ vertices of $D'$. For the initial value, let $g_0 = 0$.
	\begin{proposition}[Claim 1 \cite{Khuller}]\label{prop:g_i}
		For $i = 0, 1, \ldots, k-1$, it holds $g_{i+1} - g_i \ge \frac{1}{k}(\OPT - g_i)$.
	\end{proposition}	

	\begin{proof}[Proposition~\ref{prop:g_i}]	
		Consider the iteration of the greedy algorithm,	where vertex $v_{i+1}$ is being picked. 
		We first show that at most $g_i$ vertices of $L_1 \cup L_2$ have been already dominated. 
		Note that any vertex $w \in L_1 \cup L_2$ that has been already dominated must have been dominated by
		a vertex in $\{v_1, v_2,\ldots, v_i\}$. 
		This is because no vertex from $R$ can neighbor $w$. 
		Since $g_i =	\sum_{j=1}^i p(v_j)$ is the total profit gained so far, it follows that at most
		$g_i$ vertices from $L_1 \cup L_2$ have been dominated. 
		Hence, there are at least $\OPT - g_i$ undominated vertices in $L_1 \cup L_2$. 
		Since the $k$ vertices of $L_1$ together dominate all of these, it follows that there exists at least
		one vertex $v \in L_1$ which neighbors at least $\frac{1}{k}(\OPT-g_i)$	undominated vertices.
		Since the greedy algorithm chose to pick $v_{i+1}$ at this stage, instead of the $v$ above, it follows that $p(v_{i+1}) = g_{i+1} - g_i \ge \frac{1}{k}(\OPT-g_i)$.
		\qed\end{proof}	
	
	By solving the recurrence in Claim~\ref{prop:g_i},we get $g_i \ge (1-(1-\frac{1}{k})^i)\OPT$.
	Then, for the set $D'$, it holds 
		$
		\sum_{v \in D'} p(v) = g_{\lceil ck \rceil} 
		\ge \left(1 - \left(1-\frac{1}{k}\right)^{\lceil ck \rceil}\right)\OPT
		\ge \left(1 - \left(1-\frac{1}{k}\right)^{ck}\right)\OPT \ge \left(1 - \left(\left(1-\frac{1}{k}\right)^{k}\right)^c\right)\OPT \ge (1- e^{-c})\OPT
		$.
	Moreover, let us show that an extra $k + \lceil ck \rceil$ vertices are enough to ensure that $D'$ is connected.
	We select a subset $D'' \subseteq L_2$  of size at most $|L_3 \cap D'| \le \lceil ck \rceil$ to dominate all vertices of $D' \cap L_3$.
	Then, we ensure that all vertices are connected by simply adding all the $k$ vertices of $L_1$.
	Thus, $\hat{D} = D' \cup D'' \cup L_1$ induces a connected subgraph that contains at most $k + 2\lceil ck \rceil$ vertices.
\qed\end{proof}

We can now make use of this generalized analysis and suggest a modified algorithm, parameterized by the parameter $c$, where the Quota Steiner Tree routine is called with a quota of $(1-e^{-c})\OPT$; see Algorithm~\ref{alg:modified-greedy} below.

\begin{figure}[h]
	\scalebox{0.95}{
		\begin{algorithm}[H]
			\SetKwInOut{Input}{Input}
			\SetKwInOut{Output}{Output}
			\DontPrintSemicolon
			
			\Input{A graph $G = (V(G), E(G))$, $k \in \mathbb{N}$}
			\Output{A tree $\tilde{T}$ on at most $k$ vertices}
			
			$(D, p) \leftarrow GDS(G)$\;
			$T \leftarrow QST(G, (1-e^{-c})\OPT, p)$\;
			$\tilde{T} \leftarrow  Best_k(T, p)$\;
			
			\textbf{return} $\tilde{T}$
			
			\caption{Modified Greedy Profit Labeling Algorithm for BCDS(c)}
			\label{alg:modified-greedy}
		\end{algorithm}
	}
\end{figure}

\begin{theorem}\label{thm:general-c}
	For some constant $0 < c \le 1$, there is a $(1-e^{-c})/(\lceil 8c \rceil + 4)$ approximation for BCDS.
\end{theorem}
\begin{proof}
	By Lemma~\ref{lem:analysis} and Theorem~\ref{thm:QST}, it follows that Algorithm~\ref{alg:modified-greedy} (line 2) returns a tree of size at most $2k + 4\lceil ck \rceil \le 2k + 4(ck + 1) = (4c + 2)k + 4$ with profit at least $(1-e^{-c})\OPT$.
	For a final solution, it suffices to return a subtree of $T$, namely $T'$, of size at most $k$ which dominates the maximum number of vertices (call $Best_k(\cdot)$ in line 3 of Algorithm~\ref{alg:modified-greedy}).
	This can be done in polynomial time via dynamic programming: see section 5.2.2 in \cite{Khuller}.
	
	To prove a lower bound on the number of vertices $T'$ dominates, we decompose $T$ into a set of subtrees via iteratively removing an eligible tree from $T$.
	To do so, we apply Lemma~\ref{lem:eligible} with $p = k$.
	Like in the proof of Lemma~\ref{lem:eligible-ak}, we can prove by induction that after $i$ such removals of eligible subtrees of size at most $k$, the remaining tree has at most $|T| - i\cdot (k/2-1)$ vertices.
	For $i = \lceil 8c + 3 \rceil$, the remaining tree's size is upper bounded by $(4c + 2)k + 4 - \lceil 8c + 3 \rceil\cdot (k/2-1) \le (4c + 2)k + 4 - (8c + 3)\cdot (k/2-1) =  k/2 + 8c + 7$, which is at most $k$ for a sufficiently large choice of $k$, i.e., $k \ge 16c + 14$.
	Therefore, we can decompose $T$ into $\lceil 8c + 3 \rceil + 1 = \lceil 8c \rceil + 4$ subtrees of size at most $k$, say $T_1, T_2, \ldots, T_{\lceil 8c \rceil + 4}$.
	Then, from pigeonhole principle and our decomposition, it follows $p(T') \ge \frac{1}{\lceil 8c \rceil + 4}\sum_{i = 1}^{\lceil 8c \rceil + 4} p(T_i) \ge \frac{1}{\lceil 8c \rceil + 4}p(T) \ge \frac{1}{\lceil 8c \rceil + 4}(1-e^{-c})\OPT$.
\qed\end{proof}	

For $c = 1$, Theorem~\ref{thm:general-c} matches the approximation ratio already given in Theorem~\ref{thm:1/12}.
Since the above ratio is a function of the parameter $c$, we numerically compute its maximum value to $1/11(1-e^{-7/8})$ attained for $c = 7/8$.

\begin{corollary}\label{cor:7/8}
	There is a $1/11(1-e^{-7/8})$-approximation for BCDS.
\end{corollary}


\subsection{Inapproximability}

In this Subsection, we demonstrate a first inapproximability result for BCDS by identifying a reduction from the well known \textit{Maximum Coverage} problem.

\begin{definition}[MAX-$k$-COVER]\label{def:max-cover}
	Given a positive integer $k$ and a collection of sets $S = \{S_1, S_2, \ldots, S_m\}$,
	find a set $S' \subseteq S$, where $|S'| \le k$, which maximizes the number of covered elements $|\bigcup_{S_i \in S'} S_i|$.
\end{definition}	

\begin{theorem}[\cite{Feige,Khuller-hardness}]\label{thm:max-cover-inapprox}
	For any $\epsilon > 0$, there is no polynomial time approximation algorithm for MAX-$k$-COVER within a ratio of $(1-1/e+\epsilon)$ unless P = NP.
\end{theorem}

Let us now demonstrate a \emph{gap-preserving reduction} (Definition 10.2~\cite{Arora}) which transforms an instance of MAX-$k$-COVER, namely $\MC(S,k)$, where $S = \{S_1, S_2, \ldots, S_m\}$ to an instance of BCDS, namely $\BCDS(G,k)$, where $G = (V,E)$.
For an example illustration, see Figure~\ref{fig:inapprox-bcds}.
For each set $S_i \in S$, we include a vertex $s_i$ in $V$.
Let the union of elements in the set system $\bigcup_{S_i \in S} S_i$ be represented as $\{x_1, x_2, \ldots, x_n\}$.
For each element $x_j$, we include $q$ vertices in $V$, namely $x_{j,1}, x_{j,2}, \ldots, x_{j,q}$, where $q$ is a polynomial in $m$ ($q \ge m^2$ suffices). 
Overall, $|V| = m + qn$.
In the edge set $E$, we include edges
	 $(s_i, s_j)$, for each $i, j = 1, 2, \ldots, m$, $i \neq j$,
	and 
	$(s_i, x_{j,z})$, for each $i, j$ such that $x_j \in S_i$ and for each $z = 1, 2, \ldots, q$.
Notice the size is polynomial in the input of $\MC(S,k)$, since we get $|E| \le \binom{m}{2} + mqn$.
In Lemma~\ref{lem:inapprBCDS}, let $\MC(S, k)$, respectively $\BCDS(G, k)$, also refer to the optimal solution for the corresponding MAX-$k$-COVER, resp. BCDS, instance.

\begin{figure}[h]
	\centering
	\includegraphics[scale=0.95]{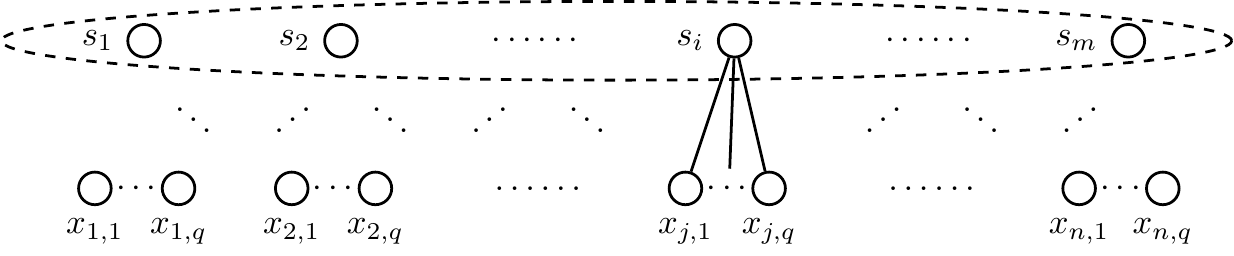}
	\caption{The graph $G$ constructed for the gap-preserving reduction employed in Lemma~\ref{lem:inapprBCDS}.
		Vertices $s_i$ within the dashed ellipse form a clique.
		Vertex $s_i$ is connected to vertices $x_{j,1}, x_{j,2}, \ldots, x_{j,q}$ in $G$ if $S_i \ni x_j$ in $\MC(S, k)$.}
	\label{fig:inapprox-bcds}
\end{figure}	

\begin{lemma}\label{lem:inapprBCDS}
	There is a gap-preserving reduction from MAX-$k$-COVER to BCDS so that,
	\begin{enumerate}
		\item[(i)] if $\MC(S,k) \ge \lambda$, then $\BCDS(G, k) \ge \Lambda$, where $\Lambda \eqdef m + q\lambda$, and
		\item[(ii)] if $\MC(S, k) < (1-\frac{1}{e}+\epsilon) \cdot \lambda$, then $\BCDS(G, k) < (1 - \frac{1}{e} + \frac{m}{e(m+q\lambda)} + \epsilon\cdot\frac{q\lambda}{m+ q\lambda})\cdot \Lambda$.
	\end{enumerate}	
\end{lemma}
\begin{proof}
If it holds $\MC(S,k) \ge \lambda$, then there exists a feasible solution $S' = \{S_{i_1}, S_{i_2}, \ldots, S_{i_l}\}$, $l \le k$, for which at least $\lambda$ elements are covered.
We form a feasible solution $C = \{s_{i_1}, s_{i_2}, \ldots, s_{i_l}\}$ for BCDS.
Since vertices $s_i$ form a clique, $\{s_{1}, s_{2}, \ldots, s_{m}\} \subseteq N[C]$.
Also, since at least $\lambda$ elements are covered, it holds that at least $q\lambda$ vertices $x_{j,z}$ are dominated by $C$; $q$ per each covered element.
Overall, it follows $\BCDS(G, k) \ge  \Lambda$.

For the second part of the proof, consider some feasible solution $D$ for BCDS.
We construct another feasible solution $D'$ as follows.
Initially, $D'$ is an empty set.
For each vertex of the form $s_i$ in $D$, add $s_i$ in $D'$.
For each vertex of the form $x_{j,z}$ in $D$, add some vertex $s_i$ in $D'$, if $s_i$ does not already exist, such that $x_j \in S_i$ in the max cover instance.
By construction $|D'| \le |D|$.
Moreover, since, for $x_j \in S_i$, it holds $N[x_{j,z}] \subseteq \{s_{1}, s_{2}, \ldots, s_{m}\} \subseteq N[s_i]$, it follows $N[D] \subseteq N[D']$.
So, for this part of the proof, it suffices to consider solutions for BCDS containing only vertices $s_i$.
Hence, if $\MC(S, k) < (1-\frac{1}{e}+\epsilon) \cdot \lambda$, at most $(1-\frac{1}{e}+\epsilon) \cdot q\lambda$ vertices $x_{j,z}$ can be dominated, and overall it holds 	

{\setstretch{1.25}
$$ \begin{array}{r l l}
\BCDS(G, k) & < m + (1-\frac{1}{e} + \epsilon)q\lambda &\\
&= (m + q\lambda) -\frac{1}{e} (m + q\lambda) + \frac{1}{e}m + \epsilon q\lambda & \\
&= (1 - \frac{1}{e} + \frac{m}{e(m+q\lambda)} + \epsilon\cdot\frac{q\lambda}{m + q\lambda})\cdot (m + q\lambda) & \\
& =(1 - \frac{1}{e} + \frac{m}{e(m+q\lambda)} + \epsilon\cdot\frac{q\lambda}{m + q\lambda})\cdot \Lambda & \\
\end{array}
$$
}
\qed\end{proof}

\begin{theorem}\label{thm:inapprBCDS}
	For any $\epsilon > 0$, there is no polynomial time approximation algorithm for BCDS within a ratio of $(1-1/e+\epsilon)$ unless P = NP.
\end{theorem}
\begin{proof}
By the second part of Lemma~\ref{lem:inapprBCDS}, if $\MC(S, k) < (1-\frac{1}{e}+\epsilon) \cdot \lambda$, then it holds $\BCDS(G, k) < (1 - \frac{1}{e} + \frac{m}{e(m+q\lambda)} + \epsilon\cdot\frac{q\lambda}{m+ q\lambda})\cdot \Lambda$.
If we select $q \ge m^\rho$ for some $\rho \ge 2$, for $m \rightarrow \infty$, it holds $\frac{m}{e(m+q\lambda)} \rightarrow 0$.
Also, we observe $0 \le \epsilon\cdot\frac{q\lambda}{m+ q\lambda} \le \epsilon$.
By combination with Theorem~\ref{thm:max-cover-inapprox}, we complete the proof.
\qed\end{proof}

\section{Edge-Vertex Domination}\label{sec:edge-vertex}

We now turn our attention to edge-vertex domination problems, where the goal is to identify a set of edges which dominate vertices of the graph. We consider both budgeted and partial cover cases.

\subsection{Budgeted Edge-Vertex Domination}\label{sec:bevd}
Let us consider the general case of BEVD (Definition~\ref{def:budget}), where the selected subset of edges does not need to be connected.
We identify a strong connection to the classical MAX-$k$-COVER problem; see Definition~\ref{def:max-cover} and Theorems~\ref{thm:max-cover-inapprox}, \ref{thm:max-cover-approx}.
On the positive side, in Theorem~\ref{thm:bevd-approx}, we prove a $(1-1/e)$-approximation by reducing BEVD to an instance of MAX-$k$-COVER.
On the negative side, we demonstrate a gap-preserving reduction from MAX-$k$-COVER to BEVD and therefore conclude that the above approximation is the best possible (Theorem~\ref{thm:inappr}).

\begin{theorem}[Proposition 5.1 \cite{Feige}]\label{thm:max-cover-approx}
	There exists a $(1-1/e)$-approximation algorithm in polynomial time for MAX-$k$-COVER.
\end{theorem}

\begin{theorem}\label{thm:bevd-approx}
	There exists a $(1-1/e)$-approximation algorithm for BEVD.
\end{theorem}
\begin{proof}
Let a graph $G = (V, E)$ and an integer $k$ be the input for BEVD.
Moreover, let $E = \{e_1, e_2, \ldots, e_m\}$.
We construct an instance $\MC(S, k)$ for MAX-$k$-COVER with input $S$ and $k$, where $S = \{S_1, S_2, \ldots, S_m\}$ and $S_i = N[e_i]$ for all $i = 1, 2, \ldots, m$.
Given a solution $S' = \{S_{i_1}, S_{i_2}, \ldots, S_{i_l}\}$ to MAX-$k$-COVER, for some $l \le k$, we transform it into a solution $E' = \{e_{i_1}, e_{i_2}, \ldots, e_{i_l}\}$ for BEVD, and vice versa.
We observe that the number of dominated vertices in $\BEVD(G,k)$ equals the number of covered elements in $\MC(S, k)$.
That is, $N[E'] = |\bigcup_{e_i \in E'} N[e_i]| = |\bigcup_{S_i \in S'} S_i|$, since by construction $N[e_i] = S_i$.
Applying Theorem~\ref{thm:max-cover-approx} completes the proof.
\qed\end{proof}

We now proceed and demonstrate a \emph{gap-preserving reduction} (Definition 10.2~\cite{Arora}) which transforms an instance of MAX-$k$-COVER, namely $\MC(S,k)$, where $S = \{S_1, S_2, \ldots, S_m\}$ to an instance of BEVD, namely $\BEVD(G,k)$, where $G = (V,E)$.
For an illustration, see Figure~\ref{fig:reduction}.
The vertex set $V$ contains a ``root'' vertex $v_0$.
For each set $S_i \in S$, we include a vertex $s_i$ in $V$.
Let the union of elements in the set system $\bigcup_{S_i \in S} S_i$ be represented as $\{x_1, x_2, \ldots, x_n\}$.
For each element $x_j$, we include $q$ vertices in $V$, namely $x_{j,1}, x_{j,2}, \ldots, x_{j,q}$, where $q$ is a polynomial in $m$ ($q \ge m^2$ suffices) 
Overall, we have $|V| = m + 1 + qn$.
In the edge set $E$, we include the edges
	$(v_0, s_i)$, for each $i = 1, 2, \ldots, m$,
	and 
	$(s_i, x_{j,z})$, for each $i, j$ such that $x_j \in S_i$ and for each $z = 1, 2, \ldots, q$.
The size is polynomial in the input of $\MC(S,k)$, since we get $|E| \le m + mqn$.
In Lemma~\ref{lem:inappr}, let $\MC(S, k)$, respectively $\BEVD(G, k)$, refer to the optimal solution for the corresponding max cover, resp.\ BEVD, instance.

\begin{figure}[h]
	\centering
	\includegraphics[scale=0.95]{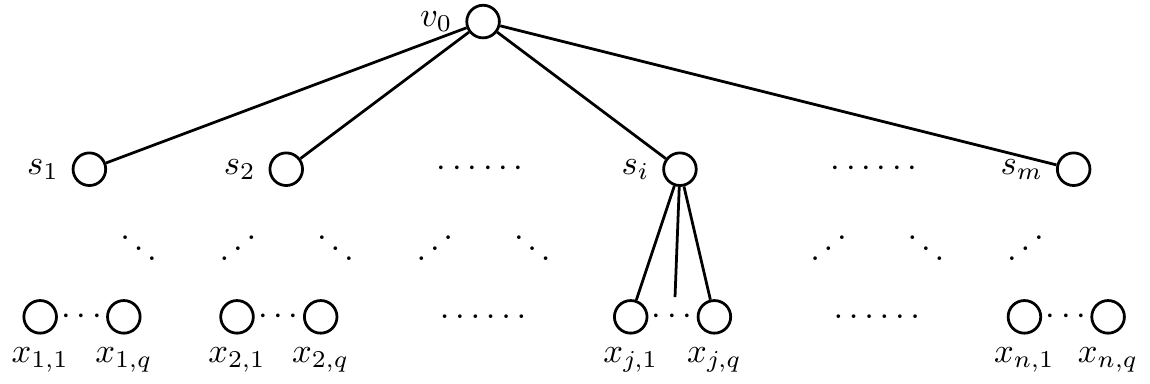}
	\caption{Graph $G$ constructed for the gap-preserving reduction employed in Lemma~\ref{lem:inappr}.
		Vertex $s_i$ is connected to vertices $x_{j,1}, x_{j,2}, \ldots, x_{j,q}$ in $G$ if $S_i \ni x_j$ in $\MC(S, k)$.}
	\label{fig:reduction}
\end{figure}	

\begin{lemma}\label{lem:inappr}
	There is a gap-preserving reduction from MAX-$k$-COVER to BEVD so that,
	\begin{itemize}
		\item[(i)] if $\MC(S, k) \ge \lambda$, then $\BEVD(G, k) \ge \Lambda$, where $\Lambda \eqdef m+1 + q\lambda$, and
		\item[(ii)] if $\MC(S, k) < (1-\frac{1}{e}+\epsilon) \cdot \lambda$, then $\BEVD(G, k) < (1 - \frac{1}{e} + \frac{m+1}{e(m+1+q\lambda)} + \epsilon\frac{q\lambda}{m+1+ q\lambda})\cdot \Lambda$.
	\end{itemize}
\end{lemma}
\begin{proof}
	If $\MC(S, k) \ge \lambda$, there exists a feasible solution $S' = \{S_{i_1}, S_{i_2}, \ldots, S_{i_l}\}$, $l \le k$, for which at least $\lambda$ elements are covered.
	We then form a feasible solution $E'$ for BEVD where $E' = \{(v_0, s_{i_1}), (v_0,s_{i_2}), \ldots, (v_0, s_{i_l})\}$.
	For each edge $(v_0, s_{i_j})$ in $E'$, it holds $ \{v_0, s_1, \ldots, s_m\} \subseteq N[(v_0,s_{i_j})]$.
	Moreover, since at least $\lambda$ elements $x_j$ are covered, then at least $q \lambda$ vertices of the form $x_{j,z}$ are dominated; $q$ per each covered element.
	It follows $\BEVD(G, k) \ge m+1+ q\lambda$.
	
	
	For part $(ii)$, consider a feasible solution for BEVD, say $E'$.
	We construct another feasible solution $E''$ as follows.
	For each edge of the form $(v_0, s_i) \in E'$, include it in $E''$.
	For each edge of the form $(s_i, x_{j,z}) \in E'$, for some $i, j, z$, include the edge $(v_0, s_i) \in E''$, if it has not been included already.
	Notice that, since at least one edge of the form $(v_0, s_i)$ is a member of $E''$, then $\{v_0, s_1, s_2, \ldots, s_m\} \subseteq N[(v_0, s_i)] \subseteq N[E'']$.
	Let us consider the differences in the number of dominated element-vertices $x_{j,z}$.
	Let $X_w = \{x_{w,1}, x_{w,2}, \ldots, x_{w,q}\}$ for $w = 1, 2, \ldots, n$
	and $Y_i = X_{i_1} \cup X_{i_2} \cup \ldots \cup X_{i_{|S_i|}}$, where $x_{i_j} \in S_i$, for $j = 1, 2, \ldots, |S_i|$, in instance $\MC(S, k)$.
	Then, it follows $Y_i \subset N[(s_i, x_{j,z})]$, since all vertices in $Y_i$ are adjacent to $s_i$.
	For the same reason, it holds $Y_i \subset N[(v_0, s_i)]$.
	Also, by construction, no other vertices $x_{j,z}$ are dominated by $(s_i, x_{j,z})$; the only ones dominated are those adjacent to $s_i$.
	The above lead us to the conclusion $N[(s_i, x_{j,z})] \subseteq N[(v_0, s_i)]$.
	Overall, inductively, we construct a feasible solution $E''$ for which $|E''| \le |E'|$ and $N[E'] \subseteq N[E'']$.
	Given the above, we need only consider one case for a feasible solution of BEVD in $G$: all selected edges are of the form $(v_0, s_i)$.
	
	We select at most $k$ edges only of the form $(v_0, s_i)$. 
	They dominate strictly fewer than $m + 1 + (1-\frac{1}{e}+\epsilon) q \lambda$ vertices.
	That is, vertices $v_0, s_1, s_2, \ldots, s_m$ are dominated and, since the selected edges are incident to at most $k$ vertices $s_i$ and by assumption $\MC(S, k) < (1-\frac{1}{e}+\epsilon) \lambda$, strictly fewer than $(1-\frac{1}{e}+\epsilon) q\lambda$ element-vertices are dominated; $q$ per each covered element.
	
	
	We obtain the following bound for the optimal solution:
	{\setstretch{1.25}
		$$ \begin{array}{r l l}
		\BEVD(G, k) & < m + 1+ (1-\frac{1}{e} + \epsilon)q\lambda &\\
		&= (m+1 + q\lambda) -\frac{1}{e} (m+1 + q\lambda) + \frac{1}{e}(m+1) + \epsilon q\lambda & \\
		&= (1 - \frac{1}{e} + \frac{m+1}{e(m+1+q\lambda)} + \epsilon\cdot\frac{q\lambda}{m+1 + q\lambda})\cdot (m+1 + q\lambda) & \\
		& =(1 - \frac{1}{e} + \frac{m+1}{e(m+1+q\lambda)} + \epsilon\cdot\frac{q\lambda}{m+1 + q\lambda})\cdot \Lambda & \\
		\end{array}
		$$
	}
	\qed\end{proof}

\begin{theorem}\label{thm:inappr}
	For any $\epsilon > 0$, there is no polynomial time approximation algorithm for BEVD within a ratio of $(1-1/e+\epsilon)$ unless P = NP.
\end{theorem}
\begin{proof}
	Notice that, in Lemma~\ref{lem:inappr}$(ii)$, if $\MC(S, k) < (1-\frac{1}{e}+\epsilon) \cdot \lambda$, it holds $\BEVD(G, k) < (1 - \frac{1}{e} + \frac{m+1}{e(m+1+q\lambda)} + \epsilon\cdot\frac{q\lambda}{m+1+ q\lambda})\cdot \Lambda$.
	For a large enough value of $q$, that is, $q \ge m^\rho$ for some $\rho \ge 2$,
	we get $\frac{m+1}{e(m+1+q\lambda)} \rightarrow 0$ as $m \rightarrow \infty$.
	Moreover, it holds $0\le \epsilon\cdot\frac{q\lambda}{m+1 + q\lambda} \le \epsilon$.
	Hence, by the combination of Theorem~\ref{thm:max-cover-inapprox} and Lemma~\ref{lem:inappr}, we complete the proof.	
\qed\end{proof}


As a side note, consider the case where the selected edge set is required to be connected.
That is, let $\BEVDC$ refer to the budgeted edge-vertex \emph{connected} domination problem.
Below, we prove that this problem is equivalent to the budgeted connected dominating set (BCDS) problem researched in Section~\ref{sec:bcds}.

\begin{proposition}\label{prop:edge-vertex-connected}
	For any $G = (V,E)$ where $|V| \ge 2$, and integer $k \ge 2$, a feasible solution $S$ to $\BCDS(G,k)$ can be transformed to a solution $S_E$ to $\BEVDC(G, k-1)$, where $N[S] = N[S_E]$, and vice versa.
\end{proposition}
\begin{proof}
Assume we are given a feasible solution $S$ to BCDS of size $|S| = s \le k$.
Since $S$ is connected, there exists a set $S_E \subseteq E$, where $|S_E| = s - 1 \le k-1$, such that $(S, S_E)$ is a tree.
Then, it holds $N[S_E] = \bigcup_{v \in S} N[v] = N[S]$, since all $v \in S$ are incident to an edge in $S_E$.

On the contrary, assume we are given a feasible solution $S_E$ to $\BEVDC(G, k-1)$ of size $|S_E| = s_E \le k-1$.
Since $S_E$ is a connected set of edges, in the best case, it is incident to a set $S$ of at most $s_E + 1 \le k$ vertices, when $(S, S_E)$ forms a tree.
In terms of neighborhood, $N[S] = \bigcup_{v \in S} N[v] = N[S_E]$, since all $v \in S$ are incident to an edge in $S_E$.
\qed\end{proof}

\subsection{Partial Edge-Vertex Domination}\label{sec:pevd}

Herein, we prove an $O(\log n)$-approximation for Partial Edge-Vertex Domination (PEVD); refer to Definition~\ref{def:partial}.
Given a graph $G=(V,E)$ and an integer $n'$, we need to select a subset $E' \subseteq E$ of minimum size such that it holds $|N[E']| \ge n'$.
To approximate the problem, we identify a reduction to \emph{Partial Cover} (PC).

\begin{definition}[PARTIAL COVER]\label{def:partial-cover}
	Given a universe (set) of elements $X = \{x_1, x_2, ..., x_n\}$, a collection of subsets of $X$, $S = \{S_1, S_2, ..., S_m\}$, and a real $0 < p \le 1$, find a minimum-size sub-collection of $S$, say $S'$, that covers at least a $p$-part of $X$, i.e., $|\bigcup_{S_i \in S'} S_i| \ge pn$.
\end{definition}

\begin{theorem}[Theorems~3, 4 in \cite{Slavik}]\label{thm:partial-cover}
	PARTIAL COVER is approximable within a factor $\min\{H(\lceil pn \rceil), H(D)\}$, where $H$ is the Harmonic number $H(x) = \sum_{i = 1}^{x} 1/x$ and $D$ is the maximum size of a set in $S$.
\end{theorem}

\begin{theorem}\label{thm:partial}
	There exists a $\min\{H(n'), H(2\Delta)\}$-approximation for PEVD.
\end{theorem}
\begin{proof}
	Given an instance $G = (V,E)$ and $n'$ of PEVD, where $|V| = n$ and $|E| = m$, we construct an instance $(S, X, p)$ of PC.
	Let $X = V$, $S_i = N[e_i]$ for $e_1, e_2, \ldots, e_m \in E$, and $p = n'/n$.
	A feasible solution for PEVD is a subset of $\rho$ edges $R = \{e_{i_1}, e_{i_2}, \ldots, e_{i_\rho}\}$.
	Equivalently, for PC, we select the corresponding collection $S_R = \{S_{i_1}, S_{i_2}, \ldots, S_{i_\rho}\}$, where $S_{i_j} = N[e_{i_j}]$.
	Notice that, for any $i_j$, it holds $|N[e_{i_j}]| \le 2(\Delta + 1) - 2 = 2\Delta$, where $\Delta$ is the maximum degree in $G$, since each endpoint of $e_{i_j}$ dominates at most $\Delta+1$ vertices; minus two overall in order not to double count the endpoints.
	Also, $N[R] = \bigcup_{e_{i_j} \in R} N[e_{i_j}] = \bigcup_{S_{i_j} \in S_R} S_{i_j}$.
	Hence, finding a solution for PEVD, which dominates at least $n'$ vertices, is equivalent to finding a cover for PC, which covers at least a $n'/n$ part of the universe $X$.
	Applying Theorem~\ref{thm:partial-cover} completes the proof.
\qed\end{proof}

\section{Conclusion}\label{sec:conclusion}

We propose a new technique to obtain tree decompositions, and a generalized analysis, thus improving the approximation guarantee from $(1-e^{-1})/13$ to $(1-e^{-7/8})/11$ for BCDS.
Furthermore, we prove a $(1-1/e+\epsilon)$ upper bound.
Also, we introduce BEVD and PEVD, and provide (tight) approximation bounds.

Regarding future work on BCDS, the goal is to design an algorithm with an improved guarantee.
Moreover, it would be interesting to capture the difficulty of the problem with a stronger inapproximability result.
We believe that a tight bound lies somewhere between our currently established state of the art.

Related to the edge-vertex case, it would be interesting to consider budgeted and partial versions for other dominating set variants, such as mixed domination \cite{Zhao}, where both vertices and edges are selected in order to dominate as many vertices and edges as possible, expansion ratio variants such as in \cite{Lamprou}, or even eternal domination \cite{Lamprou-eternal}, where a set of guards need to dominate the graph perpetually while moving to protect it against attacks on its vertices.

\end{document}